\documentclass[runningheads,envcountsect]{llncs}

\usepackage{amsmath}
\usepackage{mathrsfs}

\usepackage{multirow}
\usepackage{amsfonts}
\usepackage[colorlinks]{hyperref}
\usepackage{setspace}
\usepackage{float}
\usepackage{textcomp}
\usepackage{enumerate}
\usepackage[ruled,linesnumbered,noend]{algorithm2e}
\usepackage{subcaption}
\usepackage{marvosym}

\usepackage{graphicx}
\DeclareGraphicsExtensions{.pdf,.jpg,.png}
\newcommand{\tabincell}[2]{\begin{tabular}{@{}#1@{}}#2\end{tabular}}

\title{The Hardness of Optimization Problems on the Weighted Massively Parallel Computation Model\thanks{This work was supported by the National Natural Science Foundation of China under grants 61832003, 62273322, 61972110, and National Key Research and Development Program of China under grants 2021YFF1200100 and 2021YFF1200104.}}
\titlerunning{The hardness of optimization problems on the WMPC model}

\author{Hengzhao Ma\inst{\dag}\textsuperscript{\Letter} \and Jianzhong Li\inst{\dag}\textsuperscript{\Letter}}
\authorrunning{H. Ma, J. Li}
\institute{
	\dag\;Shenzhen Institute of Advanced Technology, Chinese Acadamy of Sciences\\
	\Letter\; \{hz.ma,lijzh\}@siat.ac.cn
}

\begin{document}
\maketitle



\begin{abstract}
The topology-aware Massively Parallel Computation (MPC) model is proposed and studied recently, which enhances the classical MPC model by the awareness of network topology. The work of Hu et al. on topology-aware MPC model considers only the tree topology. In this paper a more general case is considered, where the underlying network is a weighted complete graph. We then call this model as Weighted Massively Parallel Computation (WMPC) model, and study the problem of minimizing communication cost under it.  Two communication cost minimization problems are defined based on different pattern of communication, which are the Data Redistribution Problem and Data Allocation Problem. We also define four kinds of objective functions for communication cost, which consider the total cost, bottleneck cost, maximum of send and receive cost, and summation of send and receive cost, respectively. Combining the two problems in different communication pattern with the four kinds of objective cost functions, 8 problems are obtained. The hardness results of the 8 problems make up the content of this paper. With rigorous proof, we prove that some of the 8 problems are in P, some FPT, some NP-complete, and some W[1]-complete.
\end{abstract}




%%\received{20 February 2007}
%%\received[revised]{12 March 2009}
%%\received[accepted]{5 June 2009}


\section{Introduction}
The Massively Parallel Computation model \cite{Karloff2010}, MPC for short, has been a well acknowledged model to study parallel algorithms \cite{Dean2008,Afrati2010,AndoniNOY14,Beame2013,Beame2014,Ghaffari2018,Koutris2011,Tao2013} ever since it was proposed. Compared to other parallel computation models such as PRAM \cite{Karp1989}, BSP \cite{Valiant1990}, LogP \cite{Culler1996} and so on, the advantage of the MPC model lies in its simplicity and the power to capture the essence of computation procedure of modern share-nothing clusters. In the MPC model, computation proceeds in synchronous rounds, where  in each round the computation machines first communicate with each other, then conduct local computation. Any pair of machines can communicate in a point-to-point manner, and all the communication messages can be transferred without congestion. 

Although the MPC model is simple and powerful, one of its most important shortcomings  is revealed by some recent works \cite{Blanas2020,Hu2021}, which is the strong assumption of \textit{homogeneity}. All the machines in MPC model are considered as identical, and the communication bandwidth between any pair of machines are identical too \cite{Karloff2010}. In realistic parallel environment, the assumption of identical computation machines can be satisfied in most cases, but the assumption of identical communication bandwidth can not. Typically, a cluster consists of several racks connected by slower communication channels, and each rack includes several machines connected by faster communication channels. Thus, the communication bandwidth of in-rack and across-rack communication differ significantly, which refutes the basic assumption of homogeneous communication network in MPC model.

In order to tackle this shortcoming of the MPC model, a new \textit{topology aware massively parallel computation} model was proposed and studied in \cite{Blanas2020,Hu2021}. The computation machines are still identical in this model\footnote{There may be non-computational machines in this model, though.}, but the communication bandwidth between different pair of machines are different. This model was first proposed in recent works \cite{Blanas2020}, where the underlying communication network is represented as a graph, and the edges are assigned with a weight which represents the communication bandwidth. However, the  paper \cite{Blanas2020} only declared the new model but did not give any theoretical results. The other work \cite{Hu2021} considered three data processing tasks on this model, which are set intersection, Cartesian product and sorting. Algorithms and lower bounds about the communication cost optimization problems for the three tasks were proposed  \cite{Hu2021}. However, the authors of \cite{Hu2021} restricted the underlying communication network to trees, and the algorithm and lower bounds given in that paper can not be generalized to graphs other than trees.

In this paper, we follow the line of research started by \cite{Blanas2020,Hu2021}, and consider the topology aware massively parallel computation model in a more general case, where the underlying communication network is a complete weighted graph. In this sense, our work is a complement to the work in \cite{Hu2021}. The goal of this paper is also to minimize the communication cost. However, unlike the work in \cite{Hu2021} which considers specific computation tasks, in this paper we define general communication cost minimization problems that capture the characteristics of a variety of computation tasks.

\subsection{Description of the research problems in this paper}
\subsubsection{The WMPC Model}
We first give a more detailed description of the computational model considered in this paper, which is called Weighted Massively Parallel Computation (WMPC) model. 

In WMPC model, there are $n$ computation machines with identical computational power. The communication network is modeled as a weighted complete graph, which is represented by a $n\times n$ matrix $C$. $C$ is called the communication cost matrix from now on, and it is considered as a known parameter of the WMPC model.  $C[i,j]$ is the communication cost between computation machine $i$ and $j$ for $1\le i,j\le n$, where larger value implies larger communication cost or communication latency. $C[i,i]$ is set to 0 for $1\le i\le n$. It is assumed that all pairs of machines can communicate in a point-to-point way which is in accordance with the original MPC model, and thus $C[i,j]<\infty$ holds for $1\le i,j\le n$. The matrix $C$ is not necessary to be symmetric, i.e., $C[i,j]$ may not be equal to $C[j,i]$. 

The computation on WMPC proceeds in synchronous rounds which behaves the same with the original MPC model. In each round, the computation machines first communicate with each other, then conduct local computation. 

The initial data distribution plays an important part in the problems studied in this paper. A lot of former research works on MPC model assume that the data are uniformly split across the machines \cite{Ghaffari2018,Beame2013}. In this paper, it is assumed that the data can be arbitrarily distributed, and the amount of data placed at each machine is known in advance. This is also the same assumption adopted in \cite{Beame2014,Hu2021}.

\subsubsection{Objective functions}\label{subsubsec:cost-functions}
The goal of this paper is to minimize the communication cost under WMPC model, which is divided into \textit{send cost} and \textit{receive cost}. If  $\alpha$ amount of data is transferred from machine $i$ to machine $j$, it incurs $\alpha\cdot C[i,j]$ send cost to machine $i$, and $\alpha\cdot C[i,j]$ receive cost to machine $j$. Denote $send_i$ and $rcv_i$ to be the send and receive cost of machine $i$ for $1\le i\le n$, then we define the following four objective functions.\\
Total cost (TOTAL): $\sum_{i=1}^n send_i$.\\
Bottleneck cost (BTNK): $\max_{i=1}^n rcv_i$.\\
\underline{M}aximum of \underline{s}end and   \underline{r}eceive cost (MSR): $\max_{i=1}^n \{send_i, rcv_i\}$.\\
\underline{S}um of \underline{s}end and \underline{r}eceive cost (SSR): $\max_{i=1}^n \{send_i+ rcv_i\}$.

For the TOTAL cost function, it holds that $\sum_{i=1}^n send_i=\sum_{i=1}^n rcv_i$. For the BTNK cost function, we choose to use the bottleneck of the receive cost rather than send cost, since the receive cost also reflects the workload of local computation. The MSR cost function is used in \cite{rodiger2014locality} on the classical MPC model, and in this paper we investigate its properties on WMPC model. The SSR cost function is closely related to MSR cost function, and is basically a 2-approximation of MSR cost.

Note that the send and receive cost is defined based on the amount of data transferred between two machines. For different commutation task and different communication pattern, the way of calculating the amount of transferred data will be different. Next we will use concrete computation tasks such as parallel sorting and join as the introducing example, analyze their communication patterns, and define the problems to be studied in this paper. We will introduce two problems, named Data Redistribution Problem and Data Allocation Problem. 

\subsubsection{The Data Redistribution Problem}
Consider the following parallel sorting algorithm on classical MPC model, which is often referred as TeraSort \cite{o2008terabyte}. The algorithm first selects $n-1$ splitters $s_1\le s_2\le \cdots\le s_{n-1}$ and broadcast the splitters to all machines. The $n-1$ splitters form $n$ intervals $I_i=(s_{i-1},s_i]$ where $s_0=-\infty$ and $s_n=\infty$. After obtaining the splitters, each machine sends the local data falling in the $i$-th interval to the $i$-th machine. In such way the data is ordered across the machines. Note that the label of the machines are fixed before the algorithm starts. Then the machines conduct local sorting, and the sorting task can be finished.

Now consider running the parallel sorting algorithm on the WMPC model, and assume that the splitters have been determined. The algorithm described above asks the data in the $i$-th interval to be sent to the $i$-th machine. However, this operation may lead to non-optimal communication cost. Consider the following extreme case. The data are initially inversely sorted across the machines, i.e., for machine $i<j$, the data in machine $i$ are always no less than the data in machine $j$. In such a case, if the $i$-th interval is assigned to the $(n-i)$-th machine, there would be no need to conduct communication. However, if the algorithm asks to send the data in the $i$-th interval to the $i$-th machine, all the data will be totally redistributed, incurring large amount of communication. 

Actually, there exist two shortcomings for the above TeraSort algorithm on classical MPC model. First, it neglects the initial data distribution, and neglects the importance of the way to assign the intervals to the machines to minimize the communication cost. Second, it does not consider the difference of communication costs between different pair of machines. By tackling these two points together, the first research problem to be studied in this paper is formed, which is called the Data Redistribution Problem (DRP).

The input of DRP is two $n\times n$ matrices $T$ and $C$. $T[i,j]$ represents the amount of data in the $i$-th machine that fall in the $j$-th interval. The $C$ matrix is the communication cost matrix of the WMPC model. The output is to assign the intervals to the machines, such that the communication cost is minimized. By applying the four communication cost functions introduced in Section \ref{subsubsec:cost-functions}, we get four problems denoted as DRP-TOTAL, DRP-BTNK, DRP-MSR and DRP-SSR, respectively. The four problems are studied in Section \ref{sec:drp}.

\subsubsection{The Data Allocation Problem}

In the above case of parallel sorting, it is assumed that the splitters are known in advance. However, how to select the splitters to minimize the communication cost is also an important research problem \cite{Tao2013}, and even a new problem under the WMPC model. For a formal description, let $N$ be the total number of data records to be sorted, and $n$ be the number of machines. Under the assumption that the initial data distribution is known in advance, let $S_i=\{s_{i,1},s_{i,2},\cdots,s_{i,l_i}\}$, $1\le i\le n$, which is the data initially residing in machine $i$. $l_i$ is the number of data records in machine $i$, and $\sum_{i=1}^n l_i=N$. If the splitters are chosen as 
$s_1,s_2,\cdots, s_{n-1}$, they will form $n$ intervals $(s_{j-1},s_j]$,  where $s_0=-\infty$ and $s_n=\infty$. Let $T[i,j]=|S_i\cap (s_{j-1},s_j]|$, which is the number of data records in machine $i$ that falls into the $j$-th interval $(s_{j-1},s_j]$. To minimize the communication cost, the problem is to select $n-1$ splitters $s_1\le s_2\le \cdots\le s_{n-1}$ which split the data into $n$ intervals, then find an assignment from the intervals to the machines, such that the communication cost is minimized. This problem is called Data Allocation Problem (DAP).

\textit{Remark.} Although DRP and DAP are introduced based on sorting, they can be defined using the idea of virtual machines and physical machines. 
For DRP, the input $T[i,j]$ can be considered as the amount of data initially residing in physical machine $i$ to be processed by virtual machine $j$, and the output is a permutation which assigns virtual machines to physical machines so that the communication cost is minimized.	
For DAP, choosing the splitters can be regarded as deciding the data distribution across the virtual machines, since each virtual machine is responsible to collect the data in one interval. In such a point of view, DRP and DAP can be applied to a wide range of concrete problems. Also, DRP and DAP reflect only the problems that can be solved using one synchronous round. It will the future work to study multi-round algorithms on WMPC model.
 
\subsection{Summary of results and paper organization}
Summarizing the above descriptions, we have two kinds of problems including DRP and DAP. We also have four kinds of communication cost functions including TOTAL, BTNK, MSR and SSR. 8 problems are obtained by combining two kinds of problems with four kinds of cost functions. The hardness for the 8 problems make up the content of this paper. Table \ref{table:results} summarizes all the proposed results. Some less important results and some detailed proofs are moved to appendix due to space limitation. 

\begin{table*}[htb]
	\centering
	\caption{Summary of results}\label{table:results}
	\label{tbl:cmp-crnn}
	%\resizebox{\textwidth}{!}{
		\begin{tabular}{|c|c|c|c|c|}
			\hline
			& TOTAL & BTNK & MSR &SSR \\
			\hline
			DRP      & \tabincell{c}{P\\ (Section \ref{subsec:drp-total}) }     & \tabincell{c}{P\\ (Section \ref{subsec:drp-btnk}) }   & \tabincell{c}{NP-complete \\ (Section \ref{subsec:drp-msr}) }& \tabincell{c}{NP-complete  \\(Section \ref{subsec:drp-ssr}) }   \\
			\hline
			DAP-Cont &  \tabincell{c}{FPT\\ (Section \ref{subsec:dap-total-btnk-fpt}) }   & \tabincell{c}{FPT\\ (Section \ref{subsec:dap-total-btnk-fpt}) } & \tabincell{c}{W[1]-complete \\(Section \ref{subsubsec:dap-msr-ssr-hardness})} & \tabincell{c}{W[1]-complete \\ (Section \ref{subsubsec:dap-msr-ssr-hardness})} \\
			\hline
		
		\end{tabular}
	%}
\end{table*}

In the rest of this paper, we first introduce some denotations that will be used throughout this paper in Section \ref{subsec:denotaion}, then present the results in Section \ref{sec:drp} and \ref{sec:dap-cont} following the order given in Table \ref{table:results}. The related work and future work are delayed to Section \ref{sec:rwork} and \ref{sec:fwork}. Section \ref{sec:conc} concludes this paper.

\subsection{Denotations}\label{subsec:denotaion}
A $m\times n$ matrix $A$ is denoted as $A^{m\times n}$. The element in $A$ at row $i$ and column $j$ is denoted as $A[i,j]$. 
The set of consecutive integers $\{i,i+1,i+2,\cdots, j\}$ is denoted as $[i,j]$. 
The set of integers $\{1,2,\cdots, n\}$ is denoted as $[n]$. 

A permutation on $[n]$ is a one-to-one mapping from $[n]$ to $[n]$, and it is usually denoted as $\pi$. The set of all permutations on $[n]$ is denoted as $\Pi(n)$. Denote $\pi_i$ as the image of $i$ under $\pi$. If $\pi_i=j$, it is also said that $i$ is assigned to $j$ by permutation $\pi$.  We also use $\pi^{-1}$ to denote the inverse permutation of $\pi$, i.e., if $\pi_i=j$ then $\pi^{-1}_j=i$.
\section{The Data Redistribution Problem Series}\label{sec:drp}
\begin{definition}[DRP]\label{def:drp}
	Input: A $n\times n$ transmission matrix $T^{n\times n}$ and a $n\times n$ communication cost matrix $C^{n\times n}$, where $C[i,i]=0$ for $i\in [n]$.\\	
	Output: find a permutation $\pi\in \Pi(n)$ such that the communication cost function chosen from TOTAL, BTNK, MSR and SSR is minimized. Formally,\\
	DRP-TOTAL:
	$$\min_{\pi\in \Pi(n)} \sum_{i=1}^n\sum_{j=1}^n T[i,j]C[i,\pi_j] $$
	DRP-BTNK:
	$$\min_{\pi\in \Pi(n)}\max_{i\in [n]} {\sum\limits_{j=1}^nT[j,\pi^{-1}_i]C[j,i] }$$
	DRP-MSR:
	$$\min_{\pi\in \Pi(n)}\max_{i\in [n]} \left\{\sum\limits_{j=1}^nT[i,j]C[i,\pi_j], \sum\limits_{j=1}^nT[j,\pi^{-1}_i]C[j,i] \right\}$$
	DRP-SSR:
	$$\min_{\pi\in \Pi(n)}\max_{i\in [n]} \left\{\sum\limits_{j=1}^nT[i,j]C[i,\pi_j]+\sum\limits_{j=1}^nT[j,\pi^{-1}_i]C[j,i] \right\}$$
\end{definition}

\subsection{DRP-TOTAL}\label{subsec:drp-total}

\begin{theorem}\label{thrm:drp-total-p}
	DRP-TOTAL is equivalent to the Linear Assignment Problem (LAP) \cite{akgul1992linear}. 
\end{theorem}

\begin{proof}
	The 0-1 integral linear programming (LP) formation of DRP-TOTAL is
	\begin{equation}\label{eqtn:drp-total-lp}
		\begin{aligned}
			\min & \sum\limits_{k=1}^n\sum\limits_{i=1}^n\sum\limits_{j=1}^nT[i,j]C[i,k]x_{jk} \\
			s.t. & \sum\limits_{i=1}^n x_{ij}=1\; for\; j\in [n],\; and \;\sum\limits_{j=1}^n x_{ij}=1\; for\;  i\in [n]\\
		\end{aligned}
	\end{equation}
	We have
	\begin{equation*}\begin{aligned}
		&\sum_{k=1}^n\sum\limits_{i=1}^n\sum\limits_{j=1}^nT[i,j]C[i,k]x_{jk}=\sum\limits_{j=1}^n\sum\limits_{k=1}^n\sum\limits_{i=1}^nT[i,j]C[i,k]x_{jk}\\
		=&\sum\limits_{j=1}^n \sum\limits_{k=1}^n \left(\sum\limits_{i=1}^nT[i,j]C[i,k]\right)x_{jk}
	\end{aligned}
	\end{equation*}
	Now define another matrix $F^{n\times n}$ as $F[j,k]=\sum_{i=1}^nT[i,j]C[i,k]$, $j,k\in[n]$, and the objective function in Equation \ref{eqtn:drp-total-lp} is transformed into $\min \sum\limits_{i=1}^n\sum\limits_{j=1}^n {F[i,j]x_{ij}}$,
	which is equivalent to the linear programming formation of Linear Assignment Problem.
\qed\end{proof}

\begin{corollary}
	DRP-TOTAL can be solved in $O(n^3)$ time.
\end{corollary}
\begin{proof}
	Transforming DRP-TOTAL to LAP needs $\Theta(n^3)$ time. The Hungarian algorithm \cite{martello1987linear} to solve LAP needs $O(n^3)$ time. Then the corollary follows.
\qed\end{proof}

\subsection{DRP-BTNK}\label{subsec:drp-btnk}
\begin{theorem}\label{thrm:drp-btnk-p}
	DRP-BTNK is equivalent to the Linear Bottleneck Asssignment Problem (LBAP) \cite{burkard1999linear}.
\end{theorem}

\begin{proof}
	The 0-1 integral LP formation of DRP-BTNK is
	\begin{equation}\label{eqtn:drp-btnk-lp}
		\begin{aligned}
			\min & \max\limits_{i\in [n]}\sum\limits_{k=1}^n\sum\limits_{j=1}^n T[j,k]C[j,i]x_{ki}\\
			s.t. & \sum\limits_{i=1}^n x_{ij}=1\; for\; j\in [n],\; and \;\sum\limits_{j=1}^n x_{ij}=1\; for\;  i\in [n]\\
		\end{aligned}
	\end{equation}
	Define another matrix $F^{n\times n}$ as $F[k,i]=\sum\limits_{j=1}^n T[j,k]C[j,i]$, and the objective function in Equation \ref{eqtn:drp-btnk-lp} becomes $\min  \max\limits_{i\in [n]}\sum\limits_{k=1}^nF[k,i]x_{ki}$,
	which is equivalent to the LP formation of LBAP.
\qed\end{proof}
\begin{corollary}
	DRP-BTNK can be solved in $O(n^3)$ time.
\end{corollary}	
\begin{proof}
	Transforming DRP-BTNK to LBAP needs $\Theta(n^3)$ time. The algorithm to solve LBAP needs $O(n^3)$ time \cite{pundir2015new}. Then the corollary follows.
\qed\end{proof}

\subsection{DRP-MSR}\label{subsec:drp-msr}

\begin{theorem}\label{thrm:drp-msr-npc}
	DRP-MSR is NP-complete.
\end{theorem}

\begin{proof}
	We reduce the following PARTITION problem to DRP-MSR, which is well known to be NP-complete.\\
	Input: a set $S$ of $n$ integers $S=\{s_1,s_2,\cdots, s_n \}$, where $\sum\limits_{s_i\in S}s_i=B$.\\
	Output: decide whether there exists a partition $(S_1,S_2)$ of $S$ s.t. $S_1\cap S_2=\emptyset, S_1\cup S_2=S$, and $\sum\limits_{s_i\in S_1}s_i=\sum\limits_{s_i\in S_2}s_i=B/2$. Such partition is called \textit{perfect}.
	
	For any instance of PARTITION, construct an instance of DRP-MSR. The two matrices $T$ and $C$ as the input of DRP-MSR are of size $(2n+2)\times (2n+2)$. Let
	\begin{equation*}
		T[i,j]=\left\{\begin{aligned}
			s_j,&\; if\; (i=1\; and\; j\in [1,n])\; or\;(i=2\; and\; j\in [1,n]) \\
			0,&\; otherwise\\
		\end{aligned}
		\right.
	\end{equation*}

%\begin{equation}
%	T[i,j]=\left\{\begin{aligned}
%		s_j,&\; if\; i=1\; and\; j\in [1,n] \\
%		s_j,&\; if\; i=2\; and\; j\in [1,n] \\
%		0,&\; otherwise\\
%	\end{aligned}
%	\right.
%\end{equation}
	Let $\Delta$ be a sufficiently large integer (it suffices to set $\Delta =B$), and
	\begin{equation*}%\label{eqtn:drp-msr-reduction-matrix-c}
		C[i,j]=\left\{\begin{aligned}
			1,&\; if\; (i=1, j\in [3,n+2])\; or\; (i=2,j\in [n+3,2n+2])  \\
			\Delta, &\; if (i=1,j=2)\; or\; (i=2,j=1)\\
			0,&\; otherwise\\
		\end{aligned}
		\right.
	\end{equation*}
%\begin{equation}\label{eqtn:drp-msr-reduction-matrix-c}
%	C[i,j]=\left\{\begin{aligned}
%		1,&\; if\; i=1, j\in [3,n+2]  \\
%		1,&\; if\; i=2,j\in [n+3,2n+2]\\
%		\Delta, &\; if (i=1,j=2)\; or\; (i=2,j=1)\\
%		0,&\; otherwise\\
%	\end{aligned}
%	\right.
%\end{equation}

	Since the communication cost between machine 1 and 2 are set to $\Delta$ which is sufficiently large, and $T[1,i]\ne 0$, $T[2,i]\ne 0$ for $i\in [1,n]$, then if some $i\in [1,n]$ is assigned to machine 1 or 2 by permutation $\pi$, e.g., $\pi_i=1$, there will be $T[1,i]C[1,\pi_i]= T[1,i]\cdot \Delta$, which incurs a sufficiently large communication cost between machine 1 and 2. Thus, the integers in $[1,n]$ must be assigned to some $j\in [3,2n+2]$.
	Notice that machine 1 connects only to machines $[3,n+2]$, and machine 2 connects only to machines $[n+3,2n+2]$, with communication cost set to 1. Then if some $i\in [1,n]$ is assigned to some $j\in[3,n+2]$, it incurs $a_i$ send cost to machine 1, and 0 send cost to machine 2. On the other hand, if some $i\in [1,n]$ is assigned to some $j\in[n+3,2n+2]$, it incurs 0 send cost to machine 1, and $a_i$ send cost to machine 2. This indeed reflects the cost of the partition.
	
	Actually, for any permutation $\pi\in \Pi(2n+2)$ where $\pi_i\ne 1,2$ for $i\in [1,n]$, the send cost of machine 1 under $\pi$ is the sum of the elements in the set $\{s_i\mid  \pi_i\in [3,n+2]\}$, and the send cost of machine 2 is the sum of the elements in the set $\{s_i\mid  \pi_i\in [n+3,2n+2]\}$. It can be easily verified that the receive cost of other machines are relatively small, since each of the machines in $[3,2n+2]$ only receives one element. Furthermore, the receive cost of machine 1 and 2, and the send cost of the other machines are all 0.	
	Denote $S_{\pi,1}=\{s_i\mid  \pi_i\in [3,n+2]\}, S_{\pi,2}=\{s_i\mid  \pi_i\in [n+3,2n+2]\}$, and $(S_{\pi,1},S_{\pi,2})$ forms a partition of $S$. 
	In such way, the permutation $\pi$ corresponds to the partition $(S_{\pi,1},S_{\pi,2})$. Finally, the MSR cost of the permutation $\pi$, denoted as $MSR_\pi$, is obtained by taking maximum over the send cost of machine 1 and 2,  i.e.,  $ MSR_\pi = \max\left\{\sum\limits_{s_i\in S_{\pi,1}}s_i,\sum\limits_{s_i\in S_{\pi,2}}s_i  \right\}$.
	
	By the above discussion, we can claim that above construction is indeed a reduction from PARTITION to DRP-MSR. If the PARTITION problem admits a perfect partition, then the optimum value of the constructed DRP-MSR instance is exactly $B/2$. If there is no perfect partition for the PARTITION problem, then the optimal value of the constructed DRP-MSR instance must be greater than $B/2$. Then the NP-completeness is proved by observing DRP-MSR is truly in NP.
\qed\end{proof}

\subsection{DRP-SSR}\label{subsec:drp-ssr}

\begin{theorem}\label{thrm:drp-ssr-npc}
	DRP-SSR is NP-complete.
\end{theorem}

\begin{proof}
	The NP-completeness proof for DRP-SSR can be essentially the same with DRP-MSR. In the DRP-MSR instance constructed in the proof of Theorem \ref{thrm:drp-msr-npc}, for any machine $i$ it holds that if $send_i\ne 0$ then $rcv_i=0$, and if $rcv_i\ne 0$ then $send_i=0$. Thus, the MSR and SSR cost value is the same for that instance. 
	An alternate proof is to reduce 3-PARTITION to DRP-SSR. The 3-PARTITION problem is defined as follows.
	
	Input: $3k$ positive integer numbers $s_1,\cdots, s_{3k}$ such that $\sum\limits_{j=1}^{3k}s_j=kB$, where $B$ is a positive integer.\\
	Output: determine whether there exists a partition $S_1,\cdots,S_k$ of the $3k$ numbers, such that $|S_l|=3$ and $\sum\limits_{s_j\in S_l}{s_j}=B$ for all $ l\in [k]$.
	
	3-PARTITION is known to be strongly NP-complete \cite{garey1978strong}, i.e., there is no algorithm whose complexity is bounded by a polynomial of $k$ and $B$ that can solve it, unless P=NP.
	
	For any instance of 3-PARTITION, construct an instance of DRP-SSR with $n=3k$ as follows.
	
	Let $T[i,j]=s_j$ for $i\in [n] , j\in [3k]$. Let $T[i,j]=0$ for $i\in[2k+1,3k] ,  j\in [3k]$. By doing so we have $2k$ machines that all have all the data, and $k$ machines that have no data.
	
	Let $C[i,j]$ be as follows.
	\begin{equation*}
		C[i,j]=\left\{\begin{aligned}
			1,&\; if\; (i\equiv j)\; mod\; k\; and\; i\ne j \\
			0,&\; otherwise\\
		\end{aligned}
		\right.
	\end{equation*}
	By doing so the $3k$ machines are grouped into $k$ groups each with 3 machines. The 3 machines in each group is connected into a triangle.
	
	For any partition of the $3k$ numbers into $k$ subsets $S_i=\{s_{j_1},s_{j_2},s_{j_3}\}$, $i\in [k]$, assign the three numbers $\{s_{j_1},s_{j_2},s_{j_3}\}$ to the three machines  $i,k+i,2k+i$ which are in the $i$-th group. Assuming $s_{j_1}\ge s_{j_2}\ge s_{j_3}$, the assignment must assure that $s_{j_3}$ is assigned to machine $2k+i$ which initially have no data. $s_{j_1}$ and $s_{j_2}$ can be assigned to machine $i, k+i$ arbitrarily. Assume w.l.o.g. that $s_{j_1}$ is assigned to $i$ and $s_{j_2}$ is assigned to $k+i$. It can be verified that this assignment can achieve the smallest possible value of the SSR objective function.
	
	Under this assignment, denote $send_i$ and $rcv_i$ to be the send and receive cost of machine $i$. For $1\le i\le k$, we get	
	\[\begin{aligned}
		&send_i= s_{j_2}+s_{j_3},		&rcv_i=			&s_{j_1}, 			&send_i+rcv_i=s_{j_1}+s_{j_2}+s_{j_3}&\\
		&send_{k+i}=s_{j_1}+s_{j_3},	&rcv_{k+i}=		&s_{j_2},			&send_{k+i}+rcv_{k+i}=s_{j_1}+s_{j_2}+s_{j_3}&\\
		&send_{2k+i}=0,					& rcv_{2k+i}=	&2\cdot s_{j_3},  	&send_{2k+i}+rcv_{2k+i}=2\cdot s_{j_3}&
	\end{aligned}\]
	
	Since $s_{j_1}\ge s_{j_2}\ge s_{j_3}$, we have $2s_{j_3}\le s_{j_1}+s_{j_2}+s_{j_3}$.
	Taking maximum over $i,k+i, 2k+i$, we get the maximum value of $send+rcv$ cost in the $i$-th group as $s_{j_1}+s_{j_2}+s_{j_3}$. The final cost of DRP-SSR under this assignment is obtained by taking maximum over $i\in [k]$.
	
	If there exists a solution of 3-PARTITION, then the above constructed instance of DRP-SSR has an optimum cost of $B$, since for each subset we have $s_{j_1}+s_{j_2}+s_{j_3}=B$. 
	
	If there is no solution of 3-PARTITION, then for the above constructed instance of DRP-SSR, the objective value of SSR cost of any assignment must be greater than $B$. The reason is that there must exist some $i$ such that $send_i+rcv_i=s_{j_1}+s_{j_2}+s_{j_3}> B$.
	
	The description of the reduction from 3-PARTITION to DRP-SSR is completed, and the NP-completeness of DRP-SSR is proved.
\qed\end{proof}

\section{The Problem Series of Data Allocation Problem}\label{sec:dap-cont}
In this section we study the parameterized hardness and algorithms for the DAP-Cont problem series, parameterized by the number of machines. We will use $N$ to denote the size of the input, and $n$ to denote the number of machines.

\begin{definition}[DAP-Cont] \label{def:dap-cont}
	Input: a set $S$ of $N$ integers divided into $n$ subsets $S_1=\{s_{1,1},s_{1,2},\cdots, s_{1,l_1} \},\cdots, S_n=\{ s_{n,1},s_{n,2},\cdots, s_{n,l_n}\}$, where  $n>1$ is the number of machines, and $l_i$ is the size of $S_i$ satisfying $\sum\limits_{i=1}^{n}l_i=N$.\\
	Output: find $n-1$ integers $s^*_1,\cdots s^*_{n-1}\in S$ and a permutation $\pi\in \Pi(n)$, such that the communication cost function chosen from TOTAL, BTNK, MSR and SSR is minimized. Formally,\\
	DAP-TOTAL:
	$$\min_{s^*_1,\cdots s^*_{n-1}\in S} \min_{\pi\in \Pi(n)}\sum_{i=1}^{n}\sum_{j=1}^{n} T[i,j]C[i,\pi_j]   $$
	DAP-BTNK:
	$$\min_{s^*_1,\cdots s^*_{n-1}\in S} \min_{\pi\in \Pi(n)}\max_{i\in[n]}\sum_{j=1}^{n} T[j,\pi^{-1}_i]C[j,i]   $$
	DAP-MSR:
	$$\min_{s^*_1,\cdots s^*_{n-1}\in S} \min_{\pi\in \Pi(n)}\max_{i\in[n]} \left\{\sum\limits_{j=1}^nT[i,j]C[i,\pi_j], \sum\limits_{j=1}^nT[j,\pi^{-1}_i]C[j,i] \right\}   $$
	DAP-SSR:
	$$\min_{s^*_1,\cdots s^*_{n-1}\in S} \min_{\pi\in \Pi(n)}\max_{i\in[n]} \left\{\sum\limits_{j=1}^nT[i,j]C[i,\pi_j]+\sum\limits_{j=1}^nT[j,\pi^{-1}_i]C[j,i] \right\}    $$
	where	$T[i,j]=|S_i\cap (s^*_{j-1},s^*_j]|$ and $s^*_0=-\infty, s^*_n=\infty$. 
\end{definition}

\subsection{The splitter-graph}\label{subsec:splitter-graph}
We introduce the splitter-graph, which transforms the problem of choosing splitters to choosing a path in a special graph. Given a set $S=\{s_1, s_2, \cdots, s_N  \}$ of integers, assuming $s_1\le s_2\le \cdots\le s_N$, and a parameter $n$, construct a graph $G(V,E)$ as follows. For $i\in[n-1], j\in[N]$, construct a vertex $v_{i,j}$. Let $v_{0,0}$ be the starting vertex $-\infty$, and $v_{n,N+1}$ be the end vertex $\infty$. Let $(v_{i,j},v_{i',j'})\in E$ iff $i+1=i'$ and $j<j'$. In such way, a vertex $v_{i,j}$ represents a splitter $s_j$ placed in the $i$-th position, and a path $(-\infty, v_{1,i_1},v_{2,i_2}\cdots v_{n-1,i_{n-1}},\infty)$ represents selecting $s_{i_1},s_{i_2},\cdots s_{i_{n-1}}$ as splitters. From now on, a splitter-graph based on set $S$ with parameter $n$ will be denoted as $G_s(V,E,S,n)$. The discussions in the rest of this section will depend on the splitter-graph with different definition of edge weights. 

\begin{figure}
	\centering
	\subfloat[Demonstration of DAP-Cont.]{\includegraphics[width=0.35\linewidth]{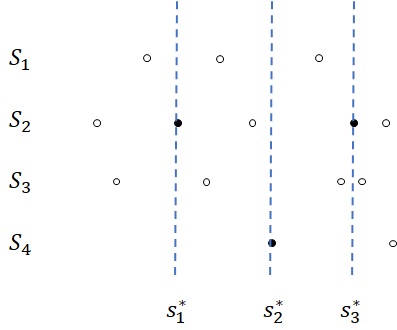}}\label{fig:dap-cont}\hspace{0.05\linewidth}
	\subfloat[The splitter-graph.]{\includegraphics[width=0.5\linewidth]{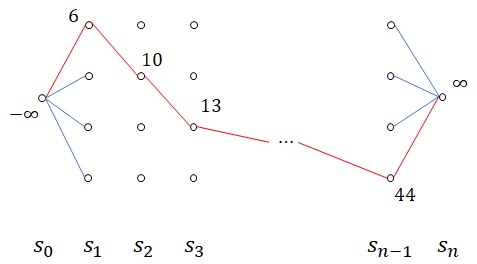}}\label{fig:splitter-graph}
\end{figure}

\subsection{FPT algorithm of DAP-TOTAL and  DAP-BTNK}\label{subsec:dap-total-btnk-fpt}
The FPT algorithm of DAP-TOTAL and  DAP-BTNK is based on the following transformation. Note that the transformation can be done in polynomial time.
Given an instance of DAP-TOTAL, denote $S=\{s_1, s_2, \cdots, s_N  \}$ and assume $s_1\le s_2\le \cdots\le s_N$. Let $s_0=-\infty$ and $s_{N+1}=\infty$. Let $Acc[i,j]=|S_i \cap (-\infty,s_j] |, i\in [n], j\in [0,N+1]$. 
Slightly abusing denotation, let $\pi^{n\times n}$ be a matrix defined based on the permutation $\pi$, such that $\pi[i,j]=1$ if $\pi_i=j$, and $\pi[i,j]=0$ otherwise, $i,j\in [n]$.

Under the above denotations, DAP-TOTAL can be transformed into the following form. 
\begin{equation}\label{eqtn:dap-transform-step-1}
		\min_{s^*_1,\cdots s^*_{n-1}\in S} \min_{\pi\in \Pi(n)}
		\sum_{i=1}^{n}\sum_{j=1}^{n}\sum_{k=1}^{n} (Acc[i,s^*_j]-Acc[i,s^*_{j-1}])C[i,k]\pi[j,k] 
\end{equation}
where $s_0^*=-\infty$ and $s_p^*=\infty$. Let $F[j,k]=\sum\limits_{i=1}^p Acc[i,j]C[i,k]$, $j\in [0,N+1],k\in[n]$, then the above equation is transformed into 
\begin{equation}\label{eqtn:dap-transform-step-2}
	\centering
	\min_{s^*_1,\cdots s^*_{n-1}\in S} \min_{\pi\in \Pi(n)}\sum_{j=1}^{n}\sum_{k=1}^{n} (F[s_j^*,k]-F[s^*_{j-1},k])\pi[j,k]
\end{equation}
Let $Cost[i,j,k]=F[i,k]-F[j,k]$, $0\le j<i\le N+1, k\in[n]$, then Equation \ref{eqtn:dap-transform-step-2} is transformed into
\begin{equation}\label{eqtn:dap-transform-step-3}
	\centering
	\min_{s^*_1,\cdots s^*_{n-1}\in S} \min_{\pi\in \Pi(n)}\sum_{j=1}^{n}\sum_{k=1}^{n} Cost[s_j^*,s_{j-1}^*,k]\pi[j,k]
\end{equation}
If $\pi$ is represented by a permutation, we get  
\begin{equation}\label{eqtn:dap-total-cost}
	\centering
	\min_{s^*_1,\cdots s^*_{p-1}\in S} \min_{\pi\in \Pi(n)}\sum_{j=1}^{n} Cost[s_j^*,s_{j-1}^*,\pi_j] 
\end{equation}

Now we can associate the above $Cost$ function to the spiltter-graph. For each edge $(v_{i,j},v_{i',j'})$ in the splitter-graph and each $l\in[n]$,  let $\omega(v_{i,j},v_{i',j'},l)=Cost[j',j,l]$.  The function $\omega$ can be regarded as assigning $p$ weights to each edge, and each weight is associated with a label  $l\in[n]$. Now, we have the following splitter-graph formation of DAP-TOTAL. 

\begin{definition}\label{def:dap-total-graph}
	Input: a splitter-graph  $G_s(V,E,S,n)$, the weight function $\omega:V\times V\times[n]\rightarrow \mathcal{R}$ of DAP-TOTAL.\\
	Output: a path $(-\infty, v_{1,i_1},v_{2,i_2}\cdots v_{n-1,i_{n-1}},\infty)$, and a permutation $\pi$, such that the following  cost is minimized 

	$$\sum\limits_{j=1}^{n}\omega(v_{j,i_j},v_{j-1,i_{j-1}},\pi_j) $$
\end{definition}
According to the above transformation, Definition \ref{def:dap-total-graph} is equivalent to the original definition of DAP-TOTAL.

\subsubsection{FPT Algorithm for Decision-DAP-TOTAL}
We prove the following decision version of DAP-TOTAL is FPT.

\begin{definition}[Decision-DAP-TOTAL]\label{def:decision-dap-total}
	Input: a splitter-graph $G_s(V,E,S,n)$, the weight function $\omega: V\times V\times[n]\rightarrow \mathcal{R}$ of DAP-TOTAL, a threshold value $\alpha$, and parameter $n$.\\
	Output: Is the optimum value of DAP-TOTAL less than $\alpha$?
\end{definition}

We need the following definition of partial permutations. A partial permutation $\pi$ is a function defined on $[i]$ where $i\in[n]$, such that $\pi_j,\pi_k\in [n]$ and $\pi_j\ne\pi_k$ for $1\le j\ne k\le i$. Here $\pi_j$ is the image of $j$ under $\pi$. Given a partial permutation $\pi$ whose definition domain is $[i]$, and an integer $l\in[n]$, let $l\in \pi$ denote that there exists some $j\in [i]$ such that $\pi_j=l$. Given an integer $l\notin \pi$, let $\pi\cup \{l\}$ be a new partial permutation $\pi'$ defined on $[i+1]$ such that $\pi'_{i+1}=l$ and $\pi'_j=\pi_j$ for $j\in [i]$. Given an integer $l=\pi_{i}$, let $\pi\setminus\{l\}$ be a partial permutation $\pi'$ defined on $[i-1]$, such that $\pi'_j=\pi_j$ for all $j\in [i-1]$. Denote $\Phi$ as the empty partial permutation.

Algorithm \ref{alg:decision-dap-total} is the FPT algorithm for Decision-DAP-TOTAL. The algorithm maintains two arrays of length $O(n!)$ for each vertex $v_{i,j}$, namely Perm($v_{i,j}$) and Cost($v_{i,j},\pi$). 
Perm ($v_{i,j}$) stores all the feasible partial permutations for the path from $-\infty$ to $v_{i,j}$, and Cost($v_{i,j},\pi$) stores the partial accumulated cost value corresponding to the partial permutation $\pi$.

\begin{algorithm}
	\caption{Decision version of DAP-TOTAL}\label{alg:decision-dap-total}
	$Perm(-\infty)\gets \{\Phi\}, Cost(-\infty,\Phi)\gets 0 $\;
	\For{$1\le i\le n$}{
		\For{$1\le j \le N$}{
			\For{$1\le k \le N$}{
				\If{edge $(v_{i-1,k},v_{i,j})$ exists}{
					\For{$1\le l \le n$}{
						\ForEach{partial permutation $\pi\in$ Perm($v_{i-1,k}$)}{\label{line:dap-total:assignment-check}
							\If{$l\notin \pi$ and $Cost(v_{i-1,k},\pi)+\omega(v_{i-1,k},v_{i,j},l)\le \alpha$}{\label{line:sum-check-dap-total}
								Add $\pi\cup \{l\}$ into $Perm(v_{i,j})$\;
								$Cost(v_{i,j},\pi\cup\{l\})\gets Cost(v_{i-1,k},\pi)+\omega(v_{i-1,k},v_{i,j},l)$;
							}
						}
					}
				}
			}
		}
	}
	Return $Yes$ if $Perm(\infty)$ is non-empty, and $NO$ otherwise.
\end{algorithm}

\begin{theorem}\label{thrm:decision-dap-total-induction}
	At the end of the $i$-th iteration of the outer-most $for$-loop in Algorithm \ref{alg:decision-dap-total}, $i\in [n]$, it holds that\\
	(1) each $\pi\in Perm(v_{i,j})$ is a partial permutation, and\\
	(2) $\pi\in Perm(v_{i,j})$ iff there exists a path $(-\infty, v_{1,j_1},\cdots ,v_{i-1,j_{i-1}}, v_{i,j_i})$ such that $\sum\limits_{k=1}^i \omega(v_{k-1,j_{k-1}},v_{k,j_k},\pi(k))\le Cost(v_{i,j},\pi)$.
\end{theorem}

\begin{proof}
		(1) Straightforward. The condition in Line \ref{line:dap-total:assignment-check}  of Algorithm \ref{alg:decision-dap-total} ensures that if $l\in \pi$ then $\pi\cup \{l\}$ will not be added into $Perm(v_{i,j})$, which ensures that each $\pi\in Perm(v_{i,j})$ is a partial permutation.
		
		(2) The proof proceeds by induction on $i$. As start point of induction where $i=0$, the path is from $-\infty$ to $-\infty$, i.e., a single vertex. Then the claim trivially holds. 
		
		Suppose the claim holds at the end of the $(i-1)$-th iteration, and consider the $i$-th iteration.	
		According to Line 9 and 10  in Algorithm \ref{alg:decision-dap-total}, $\pi\in Perm(v_{i,j})$ if and only if there is an edge $(v_{i-1,j_{i-1}},v_{i,j})$ and an integer $l=\pi_{i}$  such that  $Cost(v_{i-1,j_{i-1}},\pi\setminus\{l\})+\omega(v_{i-1,j_{i-1}},v_{i,j},l)\le \alpha$. By induction hypothesis, the partial permutation $\pi\setminus\{l\}\in Perm(v_{i-1,j_{i-1}})$ if and only if there exists a path from $-\infty$ to $v_{i-1,j_{i-1}}$, $(-\infty, v_{1,j_1},\cdots ,v_{i-1,j_{i-1}})$,  such that 
		$$\sum\limits_{k=1}^{i-1} \omega(v_{k-1,j_{k-1}},v_{k,j_k},\pi_k)\le Cost(v_{i-1,j_{i-1}},\pi\setminus\{l\})$$
		Now adding the edge $(v_{i-1,k},v_{i,j})$ to the path, and adding $l$ to the partial assignment $\pi\setminus\{l\}$, we obtain a path from $-\infty$ to $v_{i,j}$ and a partial assignment $\pi$ such that  $$\sum\limits_{k=1}^i \omega(v_{k-1,j_{k-1}},v_{k,j_k},\pi_k)\le Cost(v_{i,j},\pi)$$
		By induction, the claim is proved.
\qed\end{proof}

According to Theorem \ref{thrm:decision-dap-total-induction}, by applying the induction to the vertex $\infty$, 
it holds that $Perm(\infty)\ne\emptyset$ if and only if there exists a path $(-\infty, v_{1,j_1},\cdots ,v_{n-1,j_{n-1}}, \infty)$, such that $\sum\limits_{k=1}^n \omega(v_{k-1,j_{k-1}},v_{k,j_k},\pi_k)\le Cost(\infty,\pi)\le \alpha$, where $\pi$ is a permutation in $Perm(\infty)$. 
By the splitter-graph formation of Decision-DAP-TOTAL (Definition \ref{def:dap-total-graph}), it is equivalent to that the optimum value of DAP-TOTAL is less than $\alpha$. This completes the correctness proof of Algorithm \ref{alg:decision-dap-total}. 

\begin{theorem}
	Decision-DAP-TOTAL is FPT.
\end{theorem}
\begin{proof}
	This theorem is true since Algorithm \ref{alg:decision-dap-total} solves Decision-DAP-TOTAL in $O(N^2n!n)$ time. Note that here the number of machines $n$ is regarded as the parameter, and this complexity is polynomial in the input size $N$.
\qed\end{proof}

\subsubsection{FPT Algorithm for DAP-BTNK}
Using a transformation similar with that for DAP-TOTAL, we have the following splitter-graph formation for DAP-BTNK.

\begin{definition}\label{def:dap-btnk-graph}
	Input: a splitter-graph $G_s(V,E,S,n)$, the weight function $\omega:V\times V\times[n]\rightarrow \mathcal{R}$ of DAP-BTNK.\\
	Output: a path $(-\infty, v_{1,i_1},v_{2,i_2}\cdots v_{n-1,i_{n-1}},\infty)$, and a permutation $\pi$, such that the following cost is minimized 
	$$\max\limits_{j\in[n]}\omega(v_{j,i_j},v_{j-1,i_{j-1}},\pi_j) $$
\end{definition}

The decision version of DAP-BTNK has an extra value $\alpha$ as input, and asks whether the optimum value of DAP-BTNK is less than $\alpha$. We first propose the FPT algorithm for the decision version, which is given as Algorithm \ref{alg:decision-dap-btnk}. It needs one array for each vertex $v_{i,j}$ which is Perm($v_{i,j}$). The algorithm is similar with that for Decision-DAP-TOTAL, only changing the sum-check (Line \ref{line:sum-check-dap-total} in Algorithm \ref{alg:decision-dap-total}) to maximum check (Line \ref{line:maximum-check-dap-btnk} in Algorithm \ref{alg:decision-dap-btnk}). Thus the correctness proof of this algorithm is similar with Theorem \ref{thrm:decision-dap-total-induction}, and it is omitted.

\begin{algorithm}
	\caption{Decision version of DAP-Continuous-BTNK}\label{alg:decision-dap-btnk}
	$Perm(-\infty)\gets\Phi$\;
	\For{$1\le i\le n$}{
		\For{$1\le j \le N$}{
			\For{$1\le k \le N$}{
				\If{edge $(v_{i-1,k},v_{i,j})$ exists}{
					\For{$1\le l \le n$}{
						\ForEach{partial permutation $\pi\in$ Perm($v_{i-1,k}$)}{
							\If{$l\notin \pi$ and $\omega(v_{i-1,k},v_{i,j},l)\le \alpha$}{\label{line:maximum-check-dap-btnk}
								Add $\pi\cup \{l\}$ into $Perm(v_{i,j})$\;
							}
						}
					}
				}
			}
		}
	}
	Return $Yes$ if $Perm(\infty)$ is non-empty, and $NO$ otherwise.
\end{algorithm}

\begin{theorem}\label{thrm:dap-btnk-fpt}
	DAP-BTNK is FPT.
\end{theorem}

\begin{proof}
	We can use the algorithm for Decision-DAP-BTNK to solve DAP-BTNK. The idea is similar with the two-phase algorithm for Linear Bottleneck Assignment Problem given in Section \ref{subsec:drp-ssr}. In the first phase the algorithm chooses some possible value from the input weight function $\omega:V\times V\times[n]\rightarrow \mathcal{R}$. In the second phase Algorithm \ref{alg:decision-dap-btnk} is invoked by setting $\alpha$ as the selected weight value. Since the number of possible values of the input weight function $\omega:V\times V\times[n]\rightarrow \mathcal{R}$ is at most $N^2n$, then Algorithm \ref{alg:decision-dap-btnk} is invoked for at most $N^2n$ times. Thus, the two-phase algorithm for DAP-BTNK is still FPT.
\qed\end{proof}

\subsection{W[1]-completeness of DAP-MSR and  DAP-SSR}\label{subsubsec:dap-msr-ssr-hardness}
We first prove that the two problems are in W[1].
\begin{theorem}
	DAP-MSR and DAP-SSR are in W[1].
\end{theorem}
\begin{proof}
	The proof is based on the machine characterization of the W[1] class proposed in \cite{chen2003bounded}. The following definition and theorem must be cited from \cite{chen2003bounded} to support this proof.
	
	\begin{definition}[W-program, \cite{chen2003bounded}]\label{def:w-program}
		A nondeterministic RAM program $\mathbb{P}$ is a W-program, if there is a function $f$ and a polynomial $p$ such that for every input $(x, k)$ with $|x| = n$, the program $\mathbb{P}$ on every run\\
		(1) performs at most $f(k)\cdot p(n)$ steps;\\
		(2) at most $f(k)$ steps are nondeterministic;\\
		(3) at most the first $f(k)\cdot p(n)$ registers are used;\\
		(4) at every point of the computation the registers contain numbers $\le f(k)p(n)$;
	\end{definition}
	\begin{theorem}[\cite{chen2003bounded}]\label{thrm:machine-w1}
		Let $Q$ be a parameterized problem. Then $Q\in$ W[1] if and only if, there is a computable function $h$ and a W-program $\mathbb{P}$ deciding $Q$ such that for every run of $\mathbb{P}$ all nondeterministic steps are among the last $h(k)$ steps of the computation, where $k$ is the parameter.
	\end{theorem}

	First we should note that if term (4) in Definition \ref{def:w-program} is to be satisfied, the elements in the communication cost matrix, and the edge weights in the splitter-graph, should be bounded by $f(n)p(N)$. Under this constraint, we describe the W-program for DAP-MSR which is quite simple. \\
	(1) Transform the input to splitter-graph formation.\\
	(2) Non-deterministically guess a path with length $n$. Note that by the definition in \cite{chen2003bounded}, the nondeterministic machine can guess an integer in one nondeterministic step, rather than guess a single bit. In such way the nondeterministic machine only need to perform $n$ nondeterministic steps to guess the path.\\ 
	(3) Enumerate all the $n!$ permutations, and find the optimal permutation with the smallest MSR cost value.
	
	It is obvious that the above program (1) performs at most $O(N^2n)+O(n!)$ steps, (2) perform only $n$ nondeterministic guess steps, (3) uses $O(N^2n)$ registers to record the input, $O(n)$ registers to record the selected path, $O(n)$ registers to record the enumerated permutation, and constant registers to conduct extra numeric operations. Furthermore, after the path is guessed, the subsequent computation needs $O(n!)$ time which satisfies the condition given in Theorem \ref{thrm:machine-w1}.
\qed\end{proof}

We then prove the W[1]-hardness of the two problems.
With an idea similar with that in Section \ref{subsec:dap-total-btnk-fpt}, we first transform DAP-MSR and DAP-SSR into a splitter-graph formation. We only describe the transformation for DAP-MSR, and it is similar for the other. Using the same denotations used in Equation \ref{eqtn:dap-transform-step-1}, we have the following equivalent form for DAP-MSR.
\begin{equation}\label{eqtn:dap-msr-transform-step-1}
		\begin{aligned}
			&\min_{s^*_1,\cdots s^*_{n-1}\in S} \min_{\pi\in \Pi(n)}\max\limits_{i\in [n]}
			&\left\{\begin{aligned}
				\sum\limits_{j=1}^{n}\sum\limits_{k=1}^{n} (Acc[i,s^*_j]-Acc[i,s^*_{j-1}])C[i,k]\pi[j,k]\\
				\sum\limits_{j=1}^{n}\sum\limits_{k=1}^{n} (Acc[j,s_k^*]-Acc[j,s_{k-1}^*])C[j,i]\pi[k,i]
			\end{aligned}\right\}		
	\end{aligned}
\end{equation}

Let $V$ be a $N\times N$ matrix  where each element is a vector of length $n$, and  let $V[j,k][i]=Acc[i,j]-Acc[i,k]$, $i\in [n], j,k\in [0,N+1]$, then Equation \ref{eqtn:dap-msr-transform-step-1} is transformed into
\begin{equation*}\label{eqtn:dap-msr-transform-step-2}
	\min_{s^*_1,\cdots s^*_{p-1}\in S} \min_{\pi\in \Pi(n)}\max\limits_{i\in [n]}
	\left\{\begin{aligned}	\sum\limits_{j=1}^{p}\sum\limits_{k=1}^{p}
		V[s_j^*,s_{j-1}^*][i]C[i,k]\pi[j,k]\\
		\sum\limits_{j=1}^{p}\sum\limits_{k=1}^{p}V[s_k^*,s_{k-1}^*][j]C[j,i]\pi[k,i]    
	\end{aligned}\right\}
\end{equation*}

Next we give the following splitter-graph formation of DAP-MSR. For each edge $(v_{i,j},v_{i',j'})$ in the splitter-graph, let $\omega(v_{i,j},v_{i',j'})= V[j',j]$, i.e., each edge is associated with a vector of length $n$. In such way, each path from $-\infty$ to $\infty$ corresponds to $n$ vectors of length $n$, and can form a $n\times n$ matrix $T$. It remains to solve the DRP problem, taking this matrix $T$ and the communication cost matrix $C$ as input. The formal definition is given as follows.

\begin{definition}
	Input: a splitter-graph $G_s(V,E,S,n)$, the edge weight function $\omega: V\times V\rightarrow \mathcal{R}^n$ of DAP-MSR, the communication cost matrix $C^{n\times n}$.\\
	Output: a path $(-\infty, v_{1,k_1},v_{2,k_2}\cdots v_{n-1,k_{n-1}},\infty)$, which corresponds a matrix $T^{n\times n}$ where $T[i,j]=\omega(v_{j,k_j},v_{j-1,k_{j-1}})[i]$, and a permutation $\pi$, such that the following  MSR cost function is minimized:
	$$\max\limits_{i\in[n]} \left\{\sum\limits_{j=1}^nT[i,j]C[i,\pi_j], \sum\limits_{j=1}^nT[j,\pi^{-1}_i]C[j,i] \right\}$$
\end{definition}

To prove the W[1]-hardness of the problem, we introduce an intermediate problem called Selecting-PARTITION. The idea is to reduce $k$-clique, which is W[1]-complete, to Selecting-PARTITION, and reduce Selecting-PARTITION to DAP-MSR (and similarly to DAP-SSR).

\begin{definition}[Selecting-PARTITION]
	Input: $n$ integers $S=\{s_1,s_2,\cdots,s_n \}$, target sum value $B$, and parameter $k$.\\
	Output: decide whether there exists a set $A\subset S$ with $|A|=k$, such that $A$ is a Yes-instance of PARTITION, i.e., there exists $A_1,A_2$ such that $A_1\cap A_2=\emptyset, A_1\cup A_2=A$ and $\sum_{s_i\in A_1}s_i=\sum_{s_i\in A_2}s_i=B/2$.
\end{definition}

\begin{lemma}\label{lema:para-reduction-selecting-partition}
	There is a parameterized reduction from $k$-clique to Selecting-PARTITION.
\end{lemma}
	
\begin{proof}
	Given an instance $G=(V,E)$ of $k$-clique, construct an instance of Selecting-PARTITION with parameter $k+k(k-1)/2$ as follows. Assume w.l.o.g. that $V=[n]$, and each vertex in $V$ is labeled by an integer $i\in [n]$. Choose a prime number $q$ with $q>k$. Initially let $S=\emptyset$. For each vertex $i$, add the integer $(k-1)q^i$ into S. For each edge $(i,j)\in E$, add $q^i+q^j$ into S. The construction can be done in $O(|V|+|E|)=O(n^2)$ time.
	
	For ease of reference let $S_v=\{(k-1)q,(k-1)q^2,\cdots, (k-1)q^n$ \}, and let $S_e$ be the set of integers associated with the edges, where each integer is of the form $q^i+q^j$. It can be verified that $S_v\cap S_e=\emptyset$.
	
	If there exists a $k$-clique in $G$ where the vertexes are $i_1,i_2,\cdots, i_k$, then the vertexes and edges in the clique correspond with two sets of integers $A_1=\{(k-1)q^{i_1},\cdots, (k-1)q^{i_k} \}$ and $A_2=\{q^{i_1}+q^{i_2},\cdots, q^{i_{j-1}}+q^{i_j} \}$. It is easy to see that $\sum\limits_{s_i\in A_1}s_i=\sum\limits_{s_i\in A_2}s_i$. The set $A=A_1\cup A_2$ is of size $k+k(k-1)/2$.
	
	If there exists a set $A\subset S$ where $|A|=k+k(k-1)/2$ that is a Yes-instance of PARTITION, there are two cases. The first case is  that there is at least one integer in $A$ that is from $S_v$, say, $(k-1)q^i$ for some $i$. Then there must be $(k-1)$ integers from $S_e$ which correspond to $(k-1)$ edges incident to the vertex $i$. The $(k-1)$ edges then lead to $k-1$ different vertexes, and that is totally $k$ vertexes plus vertex $i$. They must form a $k$-clique, otherwise the corresponding integers will not be a Yes-instance of PARTITION. The second case is that all the $k+k(k-1)/2$ integers are from $S_e$. In this case the $k+k(k-1)/2$ integers corresponds to $k+k(k-1)/2$ edges which form a $(k+1)$-clique, and the $(k+1)$-clique must contain a $k$-clique.
\qed\end{proof}	
	
\begin{lemma}\label{lemma:para-reduction-select-partition-dap-msr}
	There is a parameterized reduction from Selecting-PARTITION to DAP-MSR and DAP-SSR.
\end{lemma}

\begin{proof}
	We only describe the reduction to DAP-MSR, and the proof is similar for DAP-SSR.	
	Given an instance of Selecting-PARTITION, construct a splitter-graph formation of DAP-MSR as follows. First construct a splitter-graph $G_s(V,E,S,2k+2)$. For each edge $(v_{i,j},v_{i',j'})$ in the splitter-graph, let 
	\begin{equation*}
		\omega(v_{i,j},v_{i',j'})=\left\{\begin{aligned}
			(s_{j'},s_{j'},0,\cdots,0),&\; if\; i'\in[1,k]\\
			(0,0,0,\cdots,0),&\; if\; i'\in [k+1,2k+2]
		\end{aligned}\right.
	\end{equation*}
	where the length of the vectors are $2k+2$. 
	
	By such  construction,  each path $(-\infty, v_{1,i_1},v_{2,i_2}\cdots v_{2k+1,i_{2k+1}},\infty)$ corresponds to the following matrix $T^{(2k+2)\times(2k+2)}$
	\begin{equation*}
		T[l,j]=\left\{\begin{aligned}
			s_{i_j},&\; if\; l=1\; or\; 2,\; j\in [1,k] \\
			0,&\; otherwise\\
		\end{aligned}
		\right.
	\end{equation*}
	
	Let the communication cost matrix $C^{(2k+2)\times (2k+2)}$ be as follows.
	\begin{equation*}
		C[i,j]=\left\{\begin{aligned}
			1,&\; if\; i=1,j\in [3,k+2] \\
			1,&\; if\; i=2,j\in [k+3,2k+2]\\
			\Delta, &\; if (i=1,j=2)\; or\; (i=2,j=1)\\
			0,&\; otherwise\\
		\end{aligned}
		\right.
	\end{equation*}
	Thus, any selected $k$ integers $s_1,s_2,\cdots, s_k$ from $S$ correspond to two matrices $T^{(2k+2)\times(2k+2)}$ and $C^{(2k+2)\times (2k+2)}$. We can see that the structure of the two matrices $T^{(2k+2)\times(2k+2)}$ and $C^{(2k+2)\times (2k+2)}$ are the same with that used in the proof of NP-completeness of DRP-MSR. %(Equation \ref{eqtn:drp-msr-reduction-matrix-c}).
	By the reduction from PARTITION to DRP-MSR, the $k$ integers $s_1,s_2,\cdots, s_k$ forms a Yes-instance to PARTITION if and only if the corresponding matrices $T^{(2k+2)\times(2k+2)}$ and $C^{(2k+2)\times (2k+2)}$ form an instance of DRP-MSR whose optimum cost value is $B/2$. We conclude the proof by observing the above reduction takes $O(kN^2 )$ time, and the parameter of the constructed DAP-MSR instance is $n=2k+2$, which satisfies the definition of parameterized reduction.	
\qed\end{proof}

\section{Related Works}\label{sec:rwork}
\subsection{Parallel computation models}
This paper is based on the topology-aware MPC model \cite{Blanas2020,Hu2021}, which is almost new. The history of modeling parallel computation is long though, and there exist a lot of models such as PRAM, LogP, CONGEST and so on. Based on our observation, there are three important aspects for parallel computation, which are local computation, communication and data exchange between memory hierarchies. A good model for parallel computation should consider at least one aspect in detail, and neglect the other aspects if necessary. We categorize the existing models into the following three classes according to which aspect of parallel computation that the model emphasize.  An important note is that there is no research work that consider all the three aspects in a single model by now.
\subsubsection{Models that emphasize communication}
Most models consider communication as the most important aspect of parallel computation. Some of these models even totally neglect the local computation, assuming any local computation can be done in a single unit of time. There indeed exist many important contents that is worth considering for communication, such as network topology, synchronization, latency, link capacity, routing, and so on.

A lot of early models are dedicated to model the network topology, such as hypercube network \cite{malluhi1994hierarchical} and switch network \cite{kolp1994performance}. See \cite{duncan1990survey} for a survey. PRAM \cite{Karp1989} uses a shared memory to exchange message between processors, and the strategy of exclusive/non-exclusive read/write leads to different variants of PRAM, such CREW (concurrent read, exclusive write), EREW, CRCW and so on. The processors are assumed to have arbitrarily strong computational power. PRAM is the most acknowledged model for theoretical researches in parallel computation. See \cite{harris1994survey} for a survey on PRAM simulation methods. BSP \cite{Valiant1990} model emphasizes synchronization, where the computation is divided into synchronized supersteps. LogP \cite{Culler1996} emphasizes the cost of message exchanging, where $L$ stands for latency, $o$ stands for overhead, and $g$ stands for gap. CONGEST \cite{Peleg2000} emphasizes the network topology and link capacity, which is mostly used to study graph problems \cite{DBLP:conf/wdag/AhmadiKO18,DBLP:conf/wdag/Censor-HillelKP17,DBLP:conf/opodis/GonenO17}. The communication network in CONGEST has the same topology with the graph problem which it considers, and the length of the messages between any two processors in any round is restricted to $O(\log{n})$ where $n$ is the number of processors. The Congested Clique model is an variant of CONGEST, where the communication network is a complete graph. Many graph problems were studied on Congested Clique model such as MST \cite{DBLP:conf/podc/GhaffariP16}, maximum matching \cite{DBLP:conf/wdag/Gall16}, shortest path \cite{DBLP:conf/opodis/HolzerP15}, etc. There are many theoretical research works on CONGEST and Congested Clique model, and please refer to the references in the papers cited here for more related works. The MPC model \cite{Karloff2010} is abstracted from MapReduce \cite{Dean2008}, which can also be regarded as a simplified version of BSP. Finally, the topology aware MPC model \cite{Blanas2020,Hu2021} is based on the MPC model and involves the consideration for network topology, which directly inspires this work.

\subsubsection{Models that emphasize local computation}
As we have mentioned, many models neglect the local computation and focus on communication, but it does not reduce the importance of local computation in parallel computation. If the communication cost is the only consideration, it may cause the problem of workload imbalance. Actually balancing the workload is an important aspect in the research of parallel computation which is considered in a lot of research papers from 1990s until today \cite{DBLP:journals/tkde/CheungLX02,DBLP:conf/icde/LeeC93,DBLP:journals/ijat/ShimadaS21}.

To the extent of our knowledge, the only model that emphasizes the local computation is the MapReduce model in the original form \cite{Dean2008}. The MapReduce model defines the local computation as a Map function and a Reduce function. The Map function transfers the data into key-value pairs. The Reduce function conduct the pre-defined computation on the set of key-value pairs with the same key. The communication in MapReduce model is implicit and automatically excecuted by the underlying framework, which gathers all the key-value pairs with the same key to the same machine. Therefore, the main focus of MapReduce model is how to design the Map and Reduce functions, which are both local computation functions. 

After the MapReduce model is abstracted into the MPC model \cite{Karloff2010,Lattanzi2011}, the consideration of local load is still a main focus in many research works \cite{Tao2013}, especially in the study of the join operation in database theory \cite{Beame2013,Beame2014,Koutris2016,Koutris2011,Koutris2018}.
\subsubsection{Models that emphasize memory hierarchy}
The memory hierarchy of modern processors includes CPU cache, main memory, external memory and remote storage center. While the idea of memory hierarchy is first established for single processor computation \cite{aggarwal1987model,alpern1994uniform}, it is of equal importance in parallel computation \cite{DBLP:journals/paapp/LiMR96,DBLP:journals/ijhpcn/QiaoCY04,zhang2003dram}. The initial motivation of parallel computation is to deal with the problems with size exceeding the power of a single machine. As the amount of data to be processed grows larger and larger, there must be the case that the data size is still too large for each processor even if hundreds of processors are used. Thus, it is necessary to consider the external memory and input/output complexity in parallel computation. However, as far as we know, most recent researches on MPC and Congest model assume that the data can be split fine enough such that the data for each processor can fit in the local memory, which is not the case in realistic environments dealing with massive data. We think there is a lot of research opportunities in this direction.

\subsection{Assignment problem}
The DRP and DAP are closely related to the assignment problem. As we have shown, DRP-TOTAL is equivalent to Linear Assignment Problem (LAP),  DRP-BTNK is equivalent to Linear Bottleneck Assignment Problem (LBAP), and the solution to DAP-BTNK highly relies on the solution to LBAP. We refer the readers for \cite{burkard1999linear} to a good survery of linear assignment problem series.

The Quadratic Assignment Problem (QAP) \cite{DBLP:journals/eor/LoiolaANHQ07} is harder than LAP which is NP-complete. The permutation formation of QAP is $\min\limits_{\pi\in \Pi(n)}\sum\limits_{i=1}^n\sum\limits_{j=1}^n{F[i,j]D[\pi_i,\pi_j]}$, where $F^{n\times n}$ and $D^{n\times n}$ are two input matrices. QAP can be regarded as assigning  virtual machines to physical machines, where the transmission matrix $F^{n\times n}$ is defined between virtual machines, and the communication cost matrix $D^{n\times n}$  is defined between physical machines. Comparing with the definition of DRP and DAP, where the transmission matrix is defined between physical machine and virtual machines, it can be noticed that DRP and DAP has a structure more complex than LAP but simpler than QAP. Besides the equivalence of DRP-TOTAL and LAP as well as DRP-BTNK and LBAP, another interesting fact is that the 3-PARTITION problem is used in the reduction to prove the NP-completeness DRP-SSR, while in \cite{Queyranne1986Performance} the 3-PARTITION problem is used to prove that QAP is unapproximable, i.e.,  there is no polynomial time $r$-approximate algorithm for QAP where $r<\infty$, unless P=NP.

It is worth mentioning that the ROBOT problem defined in \cite{as1998bipartite} has a similar structure with DRP. The input of ROBOT problem includes two functions $f$ and $d$, where $f$ defines the relation between the physical locations and items, and $d$ defines the distance between physical locations. The $f$ and $d$ functions has a similar structure with the $T^{n\times n}$ and $C^{n\times n}$ matrix in DRP problem. The goal of ROBOT problem is to find a TSP with minimal length, where the TSP part makes it NP-complete.

\subsection{Number partition problem}
The NP-completeness proof for DRP-MSR is based on the PARTITION problem, which is often referred as the number partition problem (NPP) in existing literature. Recall that NPP is given a set $S$ of integers  $S=\{s_1,s_2,\cdots, s_n \}$, and decide whether there exists a partition $(S_1,S_2)$ of $S$ s.t.  and $\sum\limits_{s_i\in S_1}s_i=\sum\limits_{s_i\in S_2}s_i$. NPP is one of Gary and Johnson's six basic NP-complete problems \cite{hartmanis1982computers}, and the hardness of this problem is well-known. The approximation for this problem often considers the \textit{discrepancy}, which is $|\sum\limits_{s_i\in S_1}s_i-\sum\limits_{s_i\in S_2}s_i|$. A famous polynomial time heuristic proposed by Karmarker and Karp \cite{karmarkar1982differencing} is called the differencing method, and can lead to $O(n^{-\alpha\log{n}})$ discrepancy for some positive constant $\alpha$ \cite{yakir1996differencing}.  Another line of research consider to minimize the discrepancy on randomized data \cite{borgs2001phase,lueker1998exponentially,mertens1998phase}. See \cite{mertens2006number} for a good survey on NPP. We note there are no existing work trying to design approximate algorithm for the value $\max\left\{\sum\limits_{s_i\in S_1}s_i,\sum\limits_{s_i\in S_2}s_i \right\}$ as far as we know.
%, and that is why we think DRP-MSR and DRP-SSR can not be approximated within a constant factor.

\subsection{Data redistribution}
The DRP and DAP problem series emphasizes the importance of data redistribution which have been considered by a lot of former research works. The goal of data redistribution is to minimize the communication cost while satisfying the specific requirement on the data distribution. The works in  \cite{DBLP:conf/hipc/ChengL16,polychroniou2014track,rodiger2014locality} consider the data redistribution for join operation in database, in which the two papers \cite{polychroniou2014track,rodiger2014locality} partially inspire this work. Another work \cite{kurkal2009data} considers the data redistribution in sensor networks. Knoop et al. \cite{knoop2002distribution} consider the Distribution Assignment Placement in a engineering point of view, where the abbreviation coincides with DAP in this paper.

\section{Future works}\label{sec:fwork}
Recall that the problem series of DRP and DAP are introduced using parallel sorting and join as the representing example. They reflect the communication pattern of problems that can be solved in one synchronous round, i.e., after one round of communication to redistribute the data, it suffices for all the machines to conduct local computation to finish the computation task. However, there are many problems that need multiple rounds to solve. For example, joining multiple relations can be solved using one round \cite{Beame2014} or multiple rounds \cite{Afrati2017}. Computing the graph coloring \cite{chang2019complexity}, maximum matching \cite{Ghaffari2013}, shortest path \cite{dory2021constant}, etc., must use multiple rounds. The problem will be complicated to minimize the communication cost on WMPC model with multiple rounds. We have the two following observation for future works.

First, it may not be the optimal solution to solve DRP or DAP problem for each round, since the communication of each round is correlated. 
It also involves to decide a better initial distribution so that the communication cost of the subsequent parallel computation can be reduced. We call this problem the Data Pre-distribution Problem. 
Second, in this paper it is assumed that each pair of computation machine in WMPC model can communicate in a point-to-point manner, i.e., $C[i,j]<\infty$ for all elements in the communication cost matrix $C$. Actually this assumption is set to be compatible with the one round algorithm,  i.e., the machines must be able to reach each other in one round. If multiple rounds are allowed, the assumption of point-to-point communication can be removed, i.e., there can be some element $C[i,j]=\infty$. There will be many interesting but complicated problems such as routing and congestion under the WMPC model, which are left as future work.

\section{Conclusion}\label{sec:conc}
In this paper we proposed the WMPC (Weighted Massively Parallel Computation) model based on the existing works of topology-aware Massively Parallel Computation model \cite{Blanas2020,Hu2021}. The WMPC model considers the underlying computation network as a complete weighted graph, which is a complement to the work in \cite{Hu2021} where the network topology are restricted to trees. Based on the WMPC model the DRP and DAP problem series are  defined, each representing a set of problems with the same pattern of communication. We also defined four kinds of objective functions for communication cost which are TOTAL, BTNK, MSR and SSR, and obtained 8 problems combining the four objective functions with two communication pattern problems. We studied the hardness for the 8 problems, and provided substantial theoretical results. In conclusion, this paper studied the communication minimization problem on WMPC model with a scope both deep and wide, but we must point out that the proposed results only investigated a small portion of the research area on the WMPC or topology-aware MPC model. There are a lot of meaningful problems to be studied following what was studied in this paper.

\bibliographystyle{plain}
\bibliography{library}

\begin{thebibliography}{10}

\bibitem{Afrati2017}
Foto~N. Afrati, Manas~R. Joglekar, Christopher~M. Re, Semih Salihoglu, and
  Jeffrey~D. Ullman.
\newblock {GYM: A multiround distributed join algorithm}.
\newblock {\em Leibniz International Proceedings in Informatics, LIPIcs},
  68:4:1--4:18, oct 2017.

\bibitem{Afrati2010}
Foto~N. Afrati and Jeffrey~D. Ullman.
\newblock {Optimizing joins in a map-reduce environment}.
\newblock In {\em Advances in Database Technology - EDBT 2010 - 13th
  International Conference on Extending Database Technology, Proceedings},
  pages 99--110, New York, New York, USA, 2010. ACM Press.

\bibitem{aggarwal1987model}
Alok Aggarwal, Bowen Alpern, Ashok Chandra, and Marc Snir.
\newblock A model for hierarchical memory.
\newblock In {\em Proceedings of the nineteenth annual ACM symposium on Theory
  of computing}, pages 305--314, 1987.

\bibitem{DBLP:conf/wdag/AhmadiKO18}
Mohamad Ahmadi, Fabian Kuhn, and Rotem Oshman.
\newblock Distributed approximate maximum matching in the {CONGEST} model.
\newblock In Ulrich Schmid and Josef Widder, editors, {\em 32nd International
  Symposium on Distributed Computing, {DISC} 2018, New Orleans, LA, USA,
  October 15-19, 2018}, volume 121 of {\em LIPIcs}, pages 6:1--6:17. Schloss
  Dagstuhl - Leibniz-Zentrum f{\"{u}}r Informatik, 2018.

\bibitem{akgul1992linear}
Mustafa Akg{\"u}l.
\newblock The linear assignment problem.
\newblock In {\em Combinatorial optimization}, pages 85--122. Springer, 1992.

\bibitem{alpern1994uniform}
Bowen Alpern, Larry Carter, Ephraim Feig, and Ted Selker.
\newblock The uniform memory hierarchy model of computation.
\newblock {\em Algorithmica}, 12(2):72--109, 1994.

\bibitem{AndoniNOY14}
Alexandr Andoni, Aleksandar Nikolov, Krzysztof Onak, and Grigory Yaroslavtsev.
\newblock Parallel algorithms for geometric graph problems.
\newblock In David~B. Shmoys, editor, {\em Symposium on Theory of Computing,
  {STOC} 2014, New York, NY, USA, May 31 - June 03, 2014}, pages 574--583.
  {ACM}, 2014.

\bibitem{as1998bipartite}
Andr as~Frank, Eberhard Triesch, Bernhard Korte, and Jens Vygen.
\newblock On the bipartite travelling salesman problem.
\newblock Technical report, Citeseer.

\bibitem{Beame2013}
Paul Beame, Paraschos Koutris, and Dan Suciu.
\newblock {Communication steps for parallel query processing}.
\newblock In {\em Proceedings of the ACM SIGACT-SIGMOD-SIGART Symposium on
  Principles of Database Systems}, volume~64, pages 273--284, New York, New
  York, USA, oct 2013. ACM Press.

\bibitem{Beame2014}
Paul Beame, Paraschos Koutris, and Dan Suciu.
\newblock {Skew in parallel query processing}.
\newblock In {\em Proceedings of the ACM SIGACT-SIGMOD-SIGART Symposium on
  Principles of Database Systems}, pages 212--223. Association for Computing
  Machinery, 2014.

\bibitem{Blanas2020}
Spyros Blanas.
\newblock {Topology-aware Parallel Data Processing : Models , Algorithms and
  Systems at Scale}.
\newblock In {\em 10th Annual Conference on Innovative Data Systems Research
  (CIDR ‘20)}, 2020.

\bibitem{borgs2001phase}
Christian Borgs, Jennifer Chayes, and Boris Pittel.
\newblock Phase transition and finite-size scaling for the integer partitioning
  problem.
\newblock {\em Random Structures \& Algorithms}, 19(3-4):247--288, 2001.

\bibitem{burkard1999linear}
Rainer~E Burkard and Eranda Cela.
\newblock Linear assignment problems and extensions.
\newblock In {\em Handbook of combinatorial optimization}, pages 75--149.
  Springer, 1999.

\bibitem{DBLP:conf/wdag/Censor-HillelKP17}
Keren Censor{-}Hillel, Seri Khoury, and Ami Paz.
\newblock Quadratic and near-quadratic lower bounds for the {CONGEST} model.
\newblock In Andr{\'{e}}a~W. Richa, editor, {\em 31st International Symposium
  on Distributed Computing, {DISC} 2017, October 16-20, 2017, Vienna, Austria},
  volume~91 of {\em LIPIcs}, pages 10:1--10:16. Schloss Dagstuhl -
  Leibniz-Zentrum f{\"{u}}r Informatik, 2017.

\bibitem{chang2019complexity}
Yi-Jun Chang, Manuela Fischer, Mohsen Ghaffari, Jara Uitto, and Yufan Zheng.
\newblock The complexity of ($\delta$+ 1) coloring in congested clique,
  massively parallel computation, and centralized local computation.
\newblock In {\em Proceedings of the 2019 ACM Symposium on Principles of
  Distributed Computing}, pages 471--480, 2019.

\bibitem{chen2003bounded}
Yijia Chen, J{\"o}rg Flum, and Martin Grohe.
\newblock Bounded nondeterminism and alternation in parameterized complexity
  theory.
\newblock In {\em 18th IEEE Annual Conference on Computational Complexity,
  2003. Proceedings.}, pages 13--29. IEEE, 2003.

\bibitem{DBLP:conf/hipc/ChengL16}
Long Cheng and Tao Li.
\newblock Efficient data redistribution to speedup big data analytics in large
  systems.
\newblock In {\em 23rd {IEEE} International Conference on High Performance
  Computing, HiPC 2016, Hyderabad, India, December 19-22, 2016}, pages 91--100.
  {IEEE} Computer Society, 2016.

\bibitem{DBLP:journals/tkde/CheungLX02}
David~Wai{-}Lok Cheung, Sau~Dan Lee, and Yongqiao Xiao.
\newblock Effect of data skewness and workload balance in parallel data mining.
\newblock {\em {IEEE} Trans. Knowl. Data Eng.}, 14(3):498--514, 2002.

\bibitem{Culler1996}
David~E. Culler, Richard~M. Karp, David Patterson, Abhijit Sahay, Eunice~E.
  Santos, Klaus~Erik Schauser, Ramesh Subramonian, and Thorsten von Eicken.
\newblock {LogP A Practical Model of Parallel Computation}.
\newblock {\em Communications of the ACM}, 39(11):78--85, nov 1996.

\bibitem{Dean2008}
Jeffrey Dean and Sanjay Ghemawat.
\newblock {MapReduce: simplified data processing on large clusters}.
\newblock {\em Communications of the ACM}, 51(1):107--113, jan 2008.

\bibitem{dory2021constant}
Michal Dory, Orr Fischer, Seri Khoury, and Dean Leitersdorf.
\newblock Constant-round spanners and shortest paths in congested clique and
  mpc.
\newblock In {\em Proceedings of the 2021 ACM Symposium on Principles of
  Distributed Computing}, pages 223--233, 2021.

\bibitem{duncan1990survey}
Ralph Duncan.
\newblock A survey of parallel computer architectures.
\newblock {\em Computer}, 23(2):5--16, 1990.

\bibitem{DBLP:conf/wdag/Gall16}
Fran{\c{c}}ois~Le Gall.
\newblock Further algebraic algorithms in the congested clique model and
  applications to graph-theoretic problems.
\newblock In Cyril Gavoille and David Ilcinkas, editors, {\em Distributed
  Computing - 30th International Symposium, {DISC} 2016, Paris, France,
  September 27-29, 2016. Proceedings}, volume 9888 of {\em Lecture Notes in
  Computer Science}, pages 57--70. Springer, 2016.

\bibitem{garey1978strong}
Michael~R Garey and David~S Johnson.
\newblock ``strong''np-completeness results: Motivation, examples, and
  implications.
\newblock {\em Journal of the ACM (JACM)}, 25(3):499--508, 1978.

\bibitem{Ghaffari2018}
Mohsen Ghaffari, Themis Gouleakis, Christian Konrad, Slobodan Mitrovi{\'{c}},
  and Ronitt Rubinfeld.
\newblock {Improved massively parallel computation algorithms for MIS,
  matching, and vertex cover}.
\newblock In {\em Proceedings of the Annual ACM Symposium on Principles of
  Distributed Computing}, pages 129--138, New York, NY, USA, jul 2018. ACM.

\bibitem{Ghaffari2013}
Mohsen Ghaffari and Fabian Kuhn.
\newblock {Distributed minimum cut approximation}.
\newblock In {\em Lecture Notes in Computer Science (including subseries
  Lecture Notes in Artificial Intelligence and Lecture Notes in
  Bioinformatics)}, volume 8205 LNCS, pages 1--15, 2013.

\bibitem{DBLP:conf/podc/GhaffariP16}
Mohsen Ghaffari and Merav Parter.
\newblock {MST} in log-star rounds of congested clique.
\newblock In George Giakkoupis, editor, {\em Proceedings of the 2016 {ACM}
  Symposium on Principles of Distributed Computing, {PODC} 2016, Chicago, IL,
  USA, July 25-28, 2016}, pages 19--28. {ACM}, 2016.

\bibitem{DBLP:conf/opodis/GonenO17}
Tzlil Gonen and Rotem Oshman.
\newblock Lower bounds for subgraph detection in the {CONGEST} model.
\newblock In James Aspnes, Alysson Bessani, Pascal Felber, and Jo{\~{a}}o
  Leit{\~{a}}o, editors, {\em 21st International Conference on Principles of
  Distributed Systems, {OPODIS} 2017, Lisbon, Portugal, December 18-20, 2017},
  volume~95 of {\em LIPIcs}, pages 6:1--6:16. Schloss Dagstuhl -
  Leibniz-Zentrum f{\"{u}}r Informatik, 2017.

\bibitem{harris1994survey}
Tim~J Harris.
\newblock A survey of pram simulation techniques.
\newblock {\em ACM Computing Surveys (CSUR)}, 26(2):187--206, 1994.

\bibitem{hartmanis1982computers}
Juris Hartmanis.
\newblock Computers and intractability: a guide to the theory of
  np-completeness (michael r. garey and david s. johnson).
\newblock {\em Siam Review}, 24(1):90, 1982.

\bibitem{DBLP:conf/opodis/HolzerP15}
Stephan Holzer and Nathan Pinsker.
\newblock Approximation of distances and shortest paths in the broadcast
  congest clique.
\newblock In Emmanuelle Anceaume, Christian Cachin, and Maria~Gradinariu
  Potop{-}Butucaru, editors, {\em 19th International Conference on Principles
  of Distributed Systems, {OPODIS} 2015, December 14-17, 2015, Rennes, France},
  volume~46 of {\em LIPIcs}, pages 6:1--6:16. Schloss Dagstuhl -
  Leibniz-Zentrum f{\"{u}}r Informatik, 2015.

\bibitem{Hu2021}
Xiao Hu, Paraschos Koutris, and Spyros Blanas.
\newblock {Algorithms for a topology-aware massively parallel computation
  model}.
\newblock In {\em Proceedings of the ACM SIGACT-SIGMOD-SIGART Symposium on
  Principles of Database Systems}, pages 199--214. Association for Computing
  Machinery, jun 2021.

\bibitem{Karloff2010}
Howard Karloff, Siddharth Suri, and Sergei Vassilvitskii.
\newblock {A model of computation for MapReduce}.
\newblock In {\em Proceedings of the Annual ACM-SIAM Symposium on Discrete
  Algorithms}, pages 938--948, Philadelphia, PA, jan 2010. Society for
  Industrial and Applied Mathematics.

\bibitem{karmarkar1982differencing}
Narendra Karmarkar and Richard~M Karp.
\newblock {\em The differencing method of set partitioning}.
\newblock Computer Science Division (EECS), University of California Berkeley,
  1982.

\bibitem{Karp1989}
Richard~M Karp.
\newblock {A survey of parallel algorithms for shared-memory machines}, 1988.

\bibitem{knoop2002distribution}
Jens Knoop and Eduard Mehofer.
\newblock Distribution assignment placement: Effective optimization of
  redistribution costs.
\newblock {\em IEEE Transactions on Parallel and Distributed Systems},
  13(6):628--647, 2002.

\bibitem{kolp1994performance}
Otto Kolp.
\newblock Performance estimation for a parallel system with a hierarchical
  switch network.
\newblock {\em Parallel Computing}, 20(10-11):1613--1626, 1994.

\bibitem{Koutris2016}
Paraschos Koutris, Paul Beame, and Dan Suciu.
\newblock {Worst-case optimal algorithms for parallel query processing}.
\newblock In {\em Leibniz International Proceedings in Informatics, LIPIcs},
  volume~48. Schloss Dagstuhl- Leibniz-Zentrum fur Informatik GmbH, Dagstuhl
  Publishing, mar 2016.

\bibitem{Koutris2011}
Paraschos Koutris and Dan Suciu.
\newblock Parallel evaluation of conjunctive queries.
\newblock In Maurizio Lenzerini and Thomas Schwentick, editors, {\em
  Proceedings of the 30th {ACM} {SIGMOD-SIGACT-SIGART} Symposium on Principles
  of Database Systems, {PODS} 2011, June 12-16, 2011, Athens, Greece}, pages
  223--234. {ACM}, 2011.

\bibitem{Koutris2018}
Paris Koutris, Semih Salihoglu, and Dan Suciu.
\newblock {Algorithmic aspects of parallel query processing}.
\newblock In {\em Proceedings of the ACM SIGMOD International Conference on
  Management of Data}, pages 1659--1664. Association for Computing Machinery,
  may 2018.

\bibitem{kurkal2009data}
Rohini Kurkal.
\newblock {\em Data redistribution problem in data intensive sensor networks}.
\newblock PhD thesis, Wichita State University, 2009.

\bibitem{Lattanzi2011}
Silvio Lattanzi, Benjamin Moseley, Siddharth Suri, and Sergei Vassilvitskii.
\newblock {Filtering: A method for solving graph problems in MapReduce}.
\newblock In {\em Annual ACM Symposium on Parallelism in Algorithms and
  Architectures}, pages 85--94, New York, New York, USA, 2011. ACM Press.

\bibitem{DBLP:conf/icde/LeeC93}
Chiang Lee and Zue{-}An Chang.
\newblock Workload balance and page access scheduling for parallel joins in
  shared-nothing systems.
\newblock In {\em Proceedings of the Ninth International Conference on Data
  Engineering, April 19-23, 1993, Vienna, Austria}, pages 411--418. {IEEE}
  Computer Society, 1993.

\bibitem{DBLP:journals/paapp/LiMR96}
Zhiyong Li, Peter Mills, and John~H. Reif.
\newblock Models and resource metrics for parallel and distributed computation.
\newblock {\em Parallel Algorithms Appl.}, 8(1):35--59, 1996.

\bibitem{DBLP:journals/eor/LoiolaANHQ07}
Eliane~Maria Loiola, Nair Maria~Maia de~Abreu, Paulo Oswaldo~Boaventura Netto,
  Peter Hahn, and Tania~Maia Querido.
\newblock A survey for the quadratic assignment problem.
\newblock {\em Eur. J. Oper. Res.}, 176(2):657--690, 2007.

\bibitem{lueker1998exponentially}
George~S Lueker.
\newblock Exponentially small bounds on the expected optimum of the partition
  and subset sum problems.
\newblock {\em Random Structures \& Algorithms}, 12(1):51--62, 1998.

\bibitem{malluhi1994hierarchical}
Qutaibah~M. Malluhi and Magdy~A. Bayoumi.
\newblock The hierarchical hypercube: A new interconnection topology for
  massively parallel systems.
\newblock {\em IEEE Transactions on Parallel and Distributed Systems},
  5(1):17--30, 1994.

\bibitem{martello1987linear}
Silvano Martello and Paolo Toth.
\newblock Linear assignment problems.
\newblock In {\em North-Holland Mathematics Studies}, volume 132, pages
  259--282. Elsevier, 1987.

\bibitem{mertens1998phase}
Stephan Mertens.
\newblock Phase transition in the number partitioning problem.
\newblock {\em Physical Review Letters}, 81(20):4281, 1998.

\bibitem{mertens2006number}
Stephan Mertens.
\newblock Number partitioning.
\newblock {\em Computational complexity and statistical physics}, page 125,
  2006.

\bibitem{o2008terabyte}
Owen O’malley.
\newblock Terabyte sort on apache hadoop.
\newblock {\em Yahoo, available online at: http://sortbenchmark.
  org/Yahoo-Hadoop. pdf,(May)}, pages 1--3, 2008.

\bibitem{Peleg2000}
David Peleg.
\newblock {\em {Distributed computing: a locality-sensitive approach}}.
\newblock SIAM, 2000.

\bibitem{polychroniou2014track}
Orestis Polychroniou, Rajkumar Sen, and Kenneth~A Ross.
\newblock Track join: distributed joins with minimal network traffic.
\newblock In {\em Proceedings of the 2014 ACM SIGMOD international conference
  on Management of data}, pages 1483--1494, 2014.

\bibitem{pundir2015new}
Pramendra~Singh Pundir, Sandeep~Kumar Porwal, and Brijesh~P Singh.
\newblock A new algorithm for solving linear bottleneck assignment problem.
\newblock {\em Journal of Institute of Science and Technology}, 20(2):101--102,
  2015.

\bibitem{DBLP:journals/ijhpcn/QiaoCY04}
Xiangzhen Qiao, Shuqing Chen, and Laurence~Tianruo Yang.
\newblock {HPM:} a hierarchical model for parallel computations.
\newblock {\em Int. J. High Perform. Comput. Netw.}, 1(1/2/3):117--127, 2004.

\bibitem{Queyranne1986Performance}
Maurice Queyranne.
\newblock Performance ratio of polynomial heuristics for triangle inequality
  quadratic assignment problems.
\newblock {\em Operations Research Letters}, 4(5):231--234, 1986.

\bibitem{rodiger2014locality}
Wolf R{\"o}diger, Tobias M{\"u}hlbauer, Philipp Unterbrunner, Angelika Reiser,
  Alfons Kemper, and Thomas Neumann.
\newblock Locality-sensitive operators for parallel main-memory database
  clusters.
\newblock In {\em 2014 IEEE 30th International Conference on Data Engineering},
  pages 592--603. IEEE, 2014.

\bibitem{DBLP:journals/ijat/ShimadaS21}
Takumi Shimada and Haruhiko Suwa.
\newblock Optimization of cutting tool allocation to enhance workload balance
  and total completion time in parallel-type {FMS}.
\newblock {\em Int. J. Autom. Technol.}, 15(6):804--812, 2021.

\bibitem{Tao2013}
Yufei Tao, Wenqing Lin, and Xiaokui Xiao.
\newblock {Minimal MapReduce algorithms}.
\newblock In {\em Proceedings of the ACM SIGMOD International Conference on
  Management of Data}, number June, pages 529--540, New York, New York, USA,
  2013. ACM Press.

\bibitem{Valiant1990}
Leslie~G. Valiant.
\newblock {A bridging model for parallel computation}.
\newblock {\em Communications of the ACM}, 33(8):103--111, aug 1990.

\bibitem{yakir1996differencing}
Benjamin Yakir.
\newblock The differencing algorithm ldm for partitioning: a proof of a
  conjecture of karmarkar and karp.
\newblock {\em Mathematics of Operations Research}, 21(1):85--99, 1996.

\bibitem{zhang2003dram}
Yun-Quan Zhang.
\newblock Dram (h): A parallel computation model for high performance numerical
  computing.
\newblock {\em CHINESE JOURNAL OF COMPUTERS-CHINESE EDITION-},
  26(12):1660--1670, 2003.

\end{thebibliography}
 
\appendix

\end{document}